\documentclass[12pt,letter]{article}
\usepackage{amssymb,amsfonts,mathrsfs,amsmath}
\usepackage{graphicx}
\usepackage{bbm}
\usepackage[comma,round]{natbib}
\usepackage{comment}
\usepackage{myart}
\usepackage[utf8]{inputenc}
\usepackage[T1]{fontenc}
\usepackage{makeidx}
\usepackage{caption}
\usepackage{array}
\usepackage{color,soul}
\allowdisplaybreaks

\newtheorem{lemma}{Lemma}
\newtheorem{proposition}{Proposition}
\newtheorem{corollary}{Corollary}

\newtheorem{definition}{Definition}

\begin{document}

\thispagestyle{empty}
\renewcommand{\thefootnote}{*}
\begin{center}
~\\
~\\
{\Large \sc Parallel Search for Information}\footnote[1]{We thank Andrej Zlatos for the helpful discussions regarding the proof of Proposition \ref{pr6}. We also thank Zuo-Jun (Max) Shen for helpful comments on an earlier version of this manuscript.}\\
\vspace{2.0in}
{\large \sc T. Tony Ke}\\
{\it Massachusetts Institute of Technology}\\
{\it kete@mit.edu}\\
\vspace{.3in}
{\large \sc Wenpin Tang}\\
{\it University of California, Los Angeles}\\
{\it wenpintang@math.ucla.edu}\\
\vspace{.3in}
{\large \sc J. Miguel Villas-Boas}\\
{\it University of California, Berkeley}\\
{\it villas@haas.berkeley.edu}\\
\vspace{.3in}
{\large \sc Yuming Zhang}\\
{\it University of California, Los Angeles}\\
{\it yzhangpaul@math.ucla.edu}\\
\vspace{1.5in}
{\large April 2020}
\end{center}

\newpage
\vspace{1.0in}
\begin{center}
{\large \sc Parallel Search for Information}
\vspace{.5in}

{\sc Abstract}
\end{center}
We consider the problem of a decision-maker searching for information on multiple alternatives when information is learned on all alternatives simultaneously. The decision-maker has a running cost of searching for information, and has to decide when to stop searching for information and choose one alternative. The expected payoff of each alternative evolves as a diffusion process when information is being learned.  We present necessary and sufficient conditions for the solution, establishing existence and uniqueness. We show that the optimal boundary where search is stopped (free boundary) is star-shaped, and present an asymptotic characterization of the value function and the free boundary. We show properties of how the distance between the free boundary and the diagonal varies with the number of alternatives, and how the free boundary under parallel search relates to the one under sequential search, with and without economies of scale on the search costs.

\vspace{.2in}
\noindent Keywords: \textit{Optimal Stopping, Free Boundary Problem, Search Theory, Brownian Motion}
\thispagestyle{empty}

\newpage
 \setcounter{page}{1}
 \setcounter{footnote}{0}
 \renewcommand{\thefootnote}{\arabic{footnote}}

\section{Introduction}

\quad In several situations a decision-maker (DM) has to decide how long to gain information on several alternatives simultaneously at a cost before stopping to make an adoption decision. An important aspect considered here, is that the DM gains information on all alternatives at the same time and cannot choose which alternative to gain information on---which we call \textit{parallel search}. This can be, for example, the case of a consumer trying to decide among several products in a product category and passively learning about the product category, or browsing through a web site that compares several products side by side. 

If all the alternatives have a relatively low expected payoff, the DM may decide to stop the search, and not choose any of the alternatives. If two or more alternatives have a similar and sufficiently high expected payoff, the DM may decide to continue to search for information until finding out which alternative may be the best. If the expected payoff of the best alternative is clearly higher than the expected payoff of the second best alternative, the DM may decide to stop the search process and choose the best alternative. We  characterize the solution to this problem, considering some comparative statics, and comparing it with the case in which there is sequential search for information.

The problem of the DM can be set up by a value function for the DM, which is the expected payoff for the DM going forward under the optimal policy.
We give necessary conditions for the value function in Section \ref{sc1}: it is a viscosity solution to some partial differential equation (PDE) with at most linear growth. We then show in Section \ref{sc2} that the condition derived is also sufficient by establishing the existence and uniqueness of the solution to the PDE. We obtain this result with unbounded domain, which is essential for our asymptotic results.\footnote{For similar results with a bounded domain, see, for example, \cite{PSbook}.}

 One important ingredient of the problem considered is that there is a free boundary where it is optimal to stop, and this boundary is determined by the solution to the PDE.
In Section \ref{sc3}, we show a geometric property of the free boundary: it is star-shaped with respect to the origin. Moreover, and interestingly, how much is required from the best alternative in order to stop the process is increasing in the values of the other alternatives. 

 Although it is not possible to derive closed-form expressions for the value function or the free boundary, we can study the asymptotics of the value function as well as the free boundary when the expected payoff of all alternatives is large, which is presented in Section \ref{sc4}.
We provide fine estimates of the distance from the free boundary to  the line when all alternatives have the same high expected payoff for the case of two alternatives, while for the case of more than two alternatives we show that this distance is increasing in the number of alternatives, and is at most linear in the number of alternatives.
To the best of our knowledge, this is one of the few results concerning the asymptotic geometry of the optimal stopping problem in dimension $d \ge 2$. See \cite{PSbook}, \cite{GZ10}, \cite{AJO14} for studies of optimal stopping problems for the case of two alternatives.
The main difficulty in our problem is lack of closed-form expressions for the value function. Here we rely heavily on the PDE machinery.

We also compare the stopping boundary with the boundary that results from the problem where alternatives can only be learnt sequentially---one alternative at each instance of time. If the cost of parallel search is just the cost of sequential search for one alternative times the number of alternatives (i.e., no economies of scale in the number of alternatives on which the DM is searching for information), then we can show that there is more search for information in the sequential search case, than in the parallel search case, as the sequential search for information can replicate parallel search. In this case, when the expected payoffs of the alternatives go to infinity (or equivalently, when the outside option is sufficiently low), we can also show that boundaries for search under sequential search converge to the boundaries for search under parallel search.

We also consider what happens if the DM can choose, at different costs, to gain information sequentially on one alternative at a time, or to gain information on all alternatives simultaneously, with the cost of parallel search being less than the cost of sequential search multiplied by the number of alternatives (i.e., economies of scale in search over multiple alternatives). This is an interesting case to consider as decision-makers may have a chance to get sometimes information on all alternatives at a lower cost via parallel search (for example, browsing a website that compares all alternatives, or reading a magazine with general information about the product category), but other times may choose to dive into getting information about a particular alternative via sequential search. We find in this case that, if the expected payoffs of the alternatives are high enough, then it is always optimal to do parallel search.

There is some literature on the case of learning about a single alternative in comparison to an outside option (e.g., \citealt{rob81}, \citealt{mos01}, \citealt{bra12}, \citealt{fud15a}).\footnote{The case of learning a single alternative was considered with discrete costly sequential sampling in \cite{wal45}.} 
 When there is more than one uncertain alternative the problem becomes more complex, as choosing one alternative means giving up potential high payoffs from other alternatives about which the decision maker could also learn more. This paper can then be seen as extending this literature to allow for more than one alternative, which requires the solution to a partial differential equation. Another possibility, considered in \cite{ke16},\footnote{\cite{che16} consider which type of information to collect in a Poisson-type model, when the decision maker has to choose between two alternatives, with one and only one alternative having a high payoff. See also \cite{nik18}, \cite{ke19}, and \cite{heb17}. For problems where the DM gets rewards while learning see, for example, \cite{ber96}, \cite{rad99}.} is that the DM can choose to search for information on one alternative at a time (with alternatives having independent values). That simplifies the analysis because in each region in which one alternative is searched, the value function satisfies an ordinary differential equation on the state of that alternative, keeping the states of the other alternatives fixed. Here, the value function does not satisfy that property as the states of all alternatives move simultaneously. Consequently, the value function is determined by a partial differential equation (with free boundaries) on the state of any alternative. We compare the solution in this case with the solution when the DM can choose to search for information on only one alternative at a time. We also consider what happens when the DM can choose to search for information on only one alternative, or search on all alternatives simultaneously at a higher cost, with economies of scale on the number of alternatives searched.

 The literature on financial options based on multiple assets (rainbow options) is also related to this paper (see, for example, \citealt{stu82}, \citealt{joh87}, \citealt{rub91}, \citealt{bro97}). In relation to that literature, we present a different specification related to consumer search for information and show existence of a unique solution.

The remainder of the paper is organized as follows. We present the problem in the next section, give necessary conditions for the solution, and show its existence. Section \ref{sc2} shows uniqueness of the solution, and Section \ref{sc3} shows that the optimal solution is star-shaped. Section \ref{sc4}  considers asymptotic results of the solution, compares it to the case of sequential search, and presents what happens when we can have both parallel and sequential search.

\section{The Problem and Necessary Conditions for the Solution}
\label{sc1}


\subsection{Decision-Maker Problem}

Consider a consumer, whose utility of product $i$, $U_i,$ is the sum of the utility derived from each attribute of the product. $U_i=x_i +\sum_{t=1}^T a_{it}$, where $x_i$ is the consumer's initial expected utility, and $a_{it}$ is the utility of attribute $t$ of product $i$, which is uncertain to the consumer before search.  It is also assumed that $a_{it}$ is i.i.d. across $t$ and $i$, and without loss of generality, $\mathbb{E}[a_{it}]=0$. There is an outside option which is worth	 zero.

Each time by paying a search cost $c$, the consumer checks one attribute $a_{it}$ for all products $i=1,\ldots,d$. The consumer decides when to stop searching and upon stopping which product to buy so as to maximize the expected utility. After checking $t$ attributes, the consumer's conditional expected utility of product $i$ is,
$$
	X_i(t) = \mathbb{E}_t[U_i] = x_i+\sum_{s=1}^t a_{is}.
$$
Therefore, $X_i(t)$ is a random walk, which converges to the Brownian motion $B_i^{x_i}(t)$, when we scale $a_{is}$ and the search cost $c$ proportionally to infinitesimally small and take $T$ to infinity. The problem of the consumer is to decide when to stop the process, and then choose the best alternative.

An alternative formulation of this problem has Bayesian learning with an evolving state.
Suppose that the true value of the alternatives, $\widehat{X}(t),$ follows the process $d\widehat{X}(t)= \sigma \; d B(t)$ where $\sigma$ is a diagonal matrix, with general element $\sigma_{ii}$ in the diagonal, and that the signal of $\widehat{X}(t), S(t),$ a $d$-dimension vector, follows $dS(t) = \widehat{X}(t) \; dt + y \; d \widetilde{B}(t),$ with $\widetilde{B}(t)$ being a $d$-dimensional Brownian motion independent of $B(t),$ $y$ is a diagonal matrix, with general element in the diagonal $y_{ii}.$ Suppose also that the prior of $\widehat{X}(0)$ is a normal with mean $\widehat{X}(0)$ and variance-covariance $\widehat{\rho}(0)^2,$ with $\widehat{\rho}(0)$ being a diagonal matrix, with general element in the diagonal $\widehat{\rho}_{ii}(0).$ Then, the posterior mean of $\widehat{X}(t),$ $X(t),$ follows $X_i(t) = \widehat{\rho}_{ii}(t)/y_{ii}^2\cdot d\overline{B}_i(t)$, for all $i,$ with $\overline{B}(t)$ being a $d$-dimensional Brownian motion, and $d \widehat{\rho}_{ii}(t)/dt = \sigma_{ii}^2 - \widehat{\rho}(t)^2/y_{ii}^2$ for all $i$. So, we have $\widehat{\rho}(t) \rightarrow \sigma y$ as $t \rightarrow \infty.$ Then, if $\widehat{\rho}(0) = \sigma y,$ we have that $X(t)$ is stationary, $d X_{ii}(t) = \sigma_{ii}/y_{ii}\cdot d  \overline{B}_i(t)$ for all $i$ and the analysis that follows would be done on the process $X(t).$

In both formulations of the problem, we can let the expected payoffs of the $d$ alternatives at time $t$ be $B^x = (B_1^{x_1}(t), \ldots, B_d^{x_d}(t))_{t \ge 0}$, a $d$-dimensional Brownian motion starting at $x = (x_1, \ldots, x_d)$. Each component of this Brownian motion could be the value of the alternative if the process is stopped. In the consumer learning application, this would be the expected value of that product at the time when the consumer makes the purchase decision. In a financial option application, this would be the value of the asset when the option is exercised.
Let $\mathcal{T}$ be a suitable set of stopping times with respect to the natural filtration of $B^x$.
We aim to determine the following value function,
\begin{equation}
\label{eq:value}
u(x): = \sup_{\tau \in \mathcal{T}} \mathbb{E}\left[\max \left\{ B^{x_1}_1(\tau), \ldots,B^{x_d}_d(\tau), 0 \right\}- c\tau \right],
\end{equation}
where $c > 0$ is the cost per unit time.\footnote{The problem could also be considered with time discounting. The case of the cost per unit of time could be seen as the costs of processing information when learning about different alternatives.}

\subsection{General Framework}

\quad We start with the general framework of the optimal stopping problem (\ref{eq:value}).
Let $\Omega \subset \mathbb{R}^d$ be an open domain.
Consider the following stochastic differential equation (SDE):
\begin{equation}
\label{eq:SDE}
dX^x(t) = b(X^x(t)) dt + \sigma(X^x(t)) dB(t),  
\end{equation}
where the superscript $x$ denotes $X^x(0) = x \in \Omega$.
Here $(B(t); \, t \ge 0)$ is a $d$-dimensional Brownian motion starting at $0$, $b: \mathbb{R}^n \rightarrow \mathbb{R}^n$ and $\sigma : \mathbb{R}^n \rightarrow \mathbb{R}^{n \times n}$ satisfy

\begin{itemize}
\item
{\em Lipschitz condition}: there exists $C > 0$ such that
$$|b(x) - b(y)| + |\sigma(x) - \sigma(y)| \le C |x - y|.$$
\item
{\em Linear growth condition}: there exists $K > 0$ such that
$$|b(x)| + |\sigma(x) | \le K |x|.$$
\end{itemize}

It is well known that under these conditions, the SDE (\ref{eq:SDE}) has a strong solution which is pathwise unique.
See, for example, \cite{KS}, Section 5.2, for background on strong solutions to SDEs. The vector $X(t)$ has as each element $i$ the expected utility obtained if the DM were to decide to stop the search process at time $t$ and choose alternative $i.$

 Let
\begin{equation}
J_x(\tau): = \mathbb{E} \left[\int_0^{\tau} f(X^x(s))ds + g(X^x(\tau)) \right],
\end{equation}
where $\tau$ is a stopping time, and $f$, $g$ are two continuous functions with polynomial growth, or simply Lipschitz continuous functions.
We are interested in the value function
\begin{equation}
\label{eq:valuegeneral}
u(x) = \sup_{\tau \in \mathcal{T}} J_x(\tau),
\end{equation}
where $\mathcal{T}$ is a suitable set of stopping times.
Let $\mathcal{L}$ be the infinitesimal generator of the SDE (\ref{eq:SDE}).
That is,
\begin{equation*}
\mathcal{L}h = \sum_{i = 1}^d b_i \frac{\partial h}{\partial x_i} + \frac{1}{2} \sum_{i, j = 1}^d (\sigma \sigma^T)_{ij} \frac{\partial^2 h}{\partial x_i \partial x_j},
\end{equation*}
for any suitably smooth test function $h: \mathbb{R}^n \rightarrow \mathbb{R}$.

 A standard dynamic programming argument shows that $u$ is a viscosity solution to the following partial differential equation (PDE):
\begin{equation}
\label{eq:PDE}
\min(- \mathcal{L}u - f, u -g) = 0.
\end{equation}
We state the definition of viscosity solutions (and the associated definitions of subsolutions and supersolutions) in the Appendix and we also refer readers to \cite{crandall1983viscosity, ishii1987perron,ishii1989uniqueness} and \cite{CIL} for this notion.

Equation (\ref{eq:PDE}) is known as an {\em obstacle problem}, or a {\em variational inequality} (see \citealt{frehse1972regularity,kinderlehrer1980introduction}).
It exhibits two regimes:
\begin{itemize}
\item
$- \mathcal{L} u = f$ when $u > g$,
\item
$- \mathcal{L} u \ge  f$ when $u = g$.
\end{itemize}
The set $\{x\,|\,u(x) = g(x)\}$ is called the {\em contact set}, or {\em coincidence set}.
In general, a solution $u$ to (\ref{eq:PDE}) is of class $\mathcal{C}^1$ but not $\mathcal{C}^2$,
and the regularity depends on those of $f$, $g$.
We refer to \cite{caffarelli1998obstacle} for details.\footnote{See also \cite{str15}.}

Furthermore, let $g(x)$ have at most linear growth, that is, $g(x)  \le a \sum_{i=1}^d |x_i|,$ for some $a>0,$ which is a condition satisfied by $g(x)= \max \{x_1, ..., x_d,0\},$ which is the function $g$ in our application. Let also $f(x)$ be bounded from above by a negative number, which is also satisfied in our application.  Considering the optimal stopping problem \eqref{eq:valuegeneral}, we can then obtain Lemma \ref{lem:VI+bd0}, presented in the Appendix,  charactering the value function $u$ for this general case.

\subsection{Necessary Conditions for the Optimal Strategy}

Specializing to the optimal stopping problem (\ref{eq:value}), which is the focus of the analysis in the next sections, we have
\begin{equation}
\label{eq:Lfg0}
f(x_1,\ldots, x_d) = -c  \quad \mbox{and} \quad  g(x_1,\ldots, x_d) = \max\{x_1, \ldots, x_d, 0\}.
\end{equation}

We can then get the following corollary.\footnote{In terms of the SDE (\ref{eq:SDE}) this is the case when $b=0,$ and $\sigma = I$ where $I$  is the identity matrix. Several of the results in the next section can also be obtained for the general SDE (\ref{eq:SDE}) under some conditions. This is a standard technical issue that is not central to the results presented here, and therefore not considered for ease of presentation.}
\begin{corollary}
\label{lem:VI+bd}
Let $u$ be the value function defined by (\ref{eq:value}), with $\mathcal{T}: = \{\tau \mbox{ is a stopping time}: \mathbb{E} \tau < \infty\}$.
Then $u$ is a viscosity solution to
\begin{equation}
\label{eq:VIR}
 \min \left\{-\frac{1}{2} \Delta u + c, u - g \right\} = 0,
\end{equation}
where $\Delta$ is the Laplacian operator, $\sum_{i=1}^d \partial^2/\partial x_i^2$.
Moreover, we have for some $C >0$,
\begin{equation}
\label{eq:boundmgle}
g \le u \le \sum_{i = 1}^d |x_i| + C.
\end{equation}
\end{corollary}

Corollary \ref{lem:VI+bd} asserts that the value function $u$ satisfies the PDE (\ref{eq:VIR}), with at most linear growth.
We will show in the next Section that such a solution is unique.
Once the value function $u$ is determined, then we construct an optimal strategy $\tau^{*}$ by
\begin{equation}
\label{eq:optimal}
J_x(\tau^{*} ) = u(x).
\end{equation}
More precisely, starting at a position $x \in \{u > g\}$, the search will continue until it enters the contact set:
\begin{equation}
    \tau^{*} = \inf\{t > 0: B^x \in \{ u = g\}\}.
\end{equation}

\section{Uniqueness}
\label{sc2}

\quad In this section, we prove that there exists a unique viscosity solution to the PDE (\ref{eq:VIR}).
In the sequel, let $\mathcal{B}_R$ be the ball of radius $R > 0$. Suppose $O\subseteq\mathbb{R}^d$ is open, and we write $\partial O$ as its boundary and $\overline{O} := O \cup \partial O$.
We first prove a comparison principle in bounded domains.

\begin{lemma}[Comparison principle in $\mathcal{B}_R$]
\label{lem:CP}
Assume that $u_1$ is a supersolution to (\ref{eq:VIR}), and $u_2$ is a subsolution to (\ref{eq:VIR}).
If $u_1 \ge u_2$ on $\partial \mathcal{B}_R$ for some $R>0$, then $u_1 \ge u_2$ in $\mathcal{B}_R$.
\end{lemma}
\begin{proof}
Let us consider the domain $O:=\mathcal{B}_R\cap \{u_2>g\}$. Since $u_1$ is a supersolution, $u_1(x) \ge g(x)$ for all $x \in \mathcal{B}_R$ and therefore $u_1\geq u_2$ on $\overline{\mathcal{B}}_R\backslash O$. Then apply the comparison principle (Theorem 3.3 \cite{CIL}) in $O$, we have $u_1\geq u_2$ also on $ O$ which completes the proof.
\end{proof}

Finally, we consider uniqueness and show that among continuous functions that have less than quadratic growth at infinity, the solution $u$ obtained is unique.
\begin{lemma}[Comparison principle]
\label{comparison}
Let $u_1, u_2$ be respectively a subsolution and a supersolution to (\ref{eq:VIR}) in an open subset $\Omega$ of $ \mathbb{R}^d$. Suppose there is a continuous function $h:\mathbb{R}^+\to\mathbb{R}^+$ such that
\begin{equation}
    \label{cond: less2g}
    \begin{aligned}
        &\lim_{R\to\infty}h(R)=0,\quad \text{ and for all }R\geq 1\\
        &\limsup_{|x|\geq R,x\in\Omega}\frac{\max\{u_1(x),0\}+\max\{-u_2(x),0\}}{|x|^2}\leq h(R).
    \end{aligned}
\end{equation}
Suppose $u_2\geq u_1 $ on $\partial\Omega$ (note $\partial\Omega=\emptyset$ if $\Omega=\mathbb{R}^d$).
Then we have $u_2\geq u_1$ in $\Omega$.
\end{lemma}

We prove this Lemma in the Appendix. With this Lemma, we are able to compare sub and supersolutions in $\mathbb{R}^d$ as long as condition (\ref{cond: less2g}) is satisfied.
Combined with Lemma \ref{lem:VI+bd}, we get a complete characterization of the value function $u,$ which is presented in the next proposition. A significant new result is that this characterization is obtained for unbounded domains.

\begin{proposition}
Let $u$ be the value function defined by (\ref{eq:value}), with $\mathcal{T}: = \{\tau \mbox{ is a stopping time}: \mathbb{E} \tau < \infty\}$.
Then $u$ is the unique viscosity solution to (\ref{eq:VIR}) with at most linear growth.
\end{proposition}

\begin{proof}
By Corollary \ref{lem:VI+bd}, $u$ is a solution to (\ref{eq:value}) with linear growth at infinity. By the comparison principle we know this $u$ is the unique viscosity solution to (\ref{eq:VIR}) among all continuous functions satisfying
$\lim_{| x|\to\infty}{|u(x)|}/{| x|^2}=0. $
\end{proof}

\section{Star-shapedness of the Free Boundary}
\label{sc3}

\quad Let $u$ be a solution of (\ref{eq:VIR}). 
The {\em free boundary} of $u$ is defined as the interface of the sets $\{x\,|\,u(x)>g(x)\}$ and $\{x\,|\,u(x)=g(x)\}$ which we denote by $\Gamma(u)$. 
Several regularity results of $\Gamma(u)$ can be found in \cite{caffarelli1998obstacle}.
In this paper, we are interested in the global geometric property of $\Gamma(u)$. In this section we prove the star-shapedness.

Recall ``star-shapedness'' of a subset $S\subseteq\mathbb{R}^d$: $S$ is \textit{star-shaped} if there exists a point $z$ such that for each point $s\in S$ the segment connecting $s$ and $z$ lies entirely within $S$.
We say that the free boundary $\Gamma(u)=\partial\{u>g\}$ is star-shaped with respect to the origin if the set $\{u>g\}$ is star-shaped with $z=0$. The star-shapedness property of a set rules out holes in the set.

\begin{proposition}
Let $u$ be a solution to (\ref{eq:VIR}). The free boundary $\Gamma(u)$ is star-shaped with respect to the origin.
\end{proposition}
\begin{proof}
To prove star-shapedness, we only need to show that if $u(x)=g(x)$ for some $x\in\mathbb{R}^d$, then $u(tx)=g(tx)$ holds for all $t\geq 1$.

Let $v(x):=\frac{1}{t}u(tx)$.
We first show that $v$ is a subsolution to (\ref{eq:VIR}). In fact, for any $x\in\mathbb{R}^d$, if $v(x)> g(x)$ then \[u(tx)>tg(x)=t\max\{x_1,...x_d,0\}=g(tx).\] Thus,
\begin{equation}\label{P.e.2}
    -\frac{1}{2}(\Delta u)(tx)\leq -c\quad\text{ in the viscosity sense.}
\end{equation}

To show that $-\frac{1}{2}\Delta v(x)\leq -c$ in the viscosity sense, take any $\varphi\in\mathcal{C}^2$ that touches $v$ at $x$ from above. Then $\varphi^t(\cdot):=t\varphi(\cdot/t)$ touches $u$ at $tx$ from above. It follows from \eqref{P.e.2} that 
\[-\frac{1}{2}(\Delta \varphi^t)(tx)=-\frac{1}{2t}\Delta \varphi(x)\leq -c,\]
which implies $-\frac{1}{2}\Delta \varphi(x)\leq -tc\leq -c$. Therefore
$-\frac{1}{2}\Delta v(x)\leq -c$ in the viscosity sense.
So we conclude that $v$ is a subsolution.
Now take $x^0\in\mathbb{R}^d$ such that $u(x^0)=g(x^0)$.
From the order of $u$ and $v$, we get
\[u(tx^0)\leq t u(x^0)=t g(x^0)=g(tx^0).\]
On the other hand, $u(tx^0)\geq g(tx^0)$ by definition, so we must have $u(tx^0)=g(tx^0)$.
\end{proof}

 Figure \ref{fig:parallel} shows the continuation and stopping regions, as well as the free boundary separating them for the case of $d=2$.  The figure illustrates the star-shapedness of the free boundaries.
\begin{figure}[h]
\begin{center}
\includegraphics[width=0.5\textwidth]{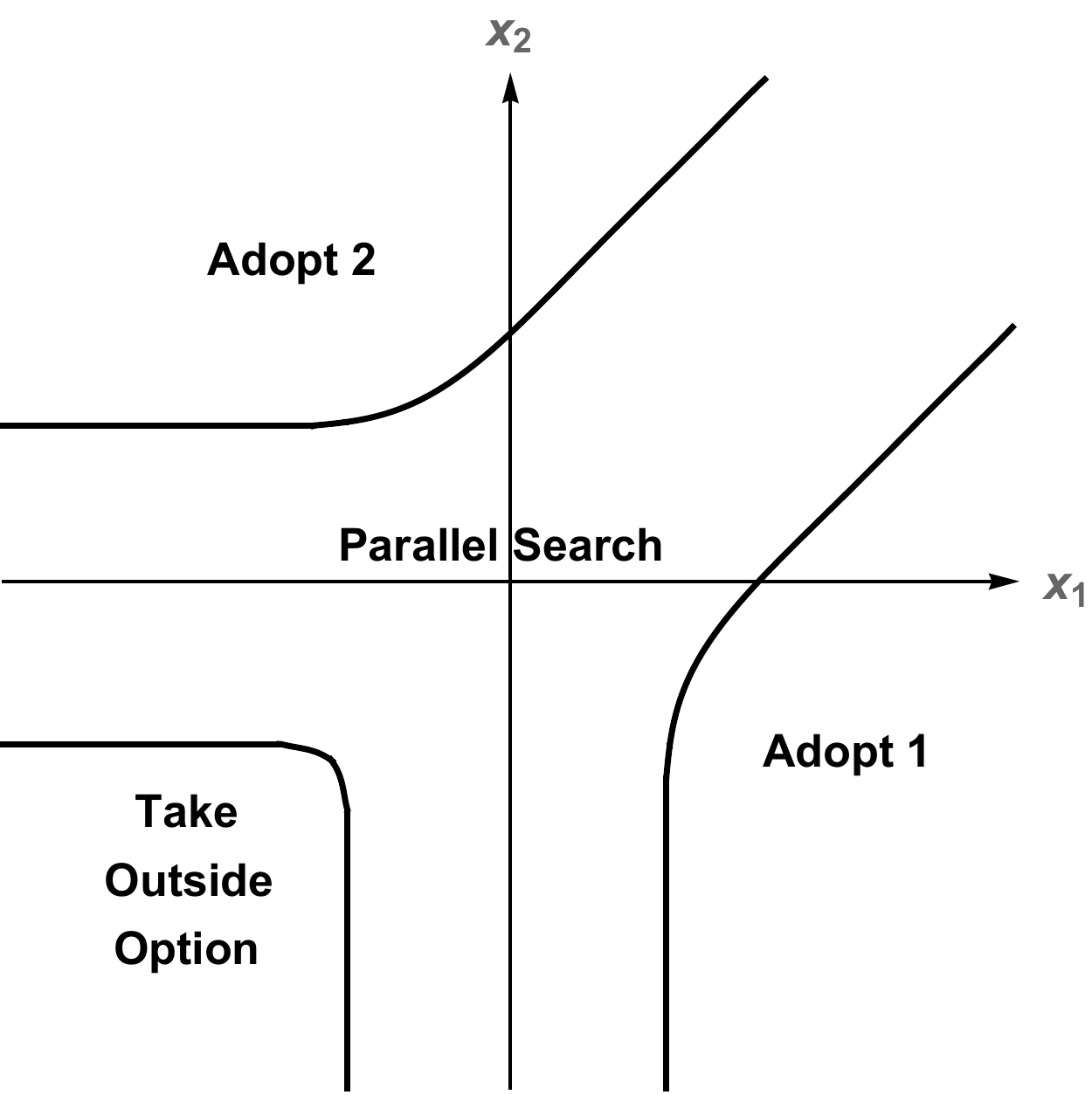}
\caption{Optimal parallel search strategy in two dimensions.}
\label{fig:parallel}
\end{center}
\end{figure}

As shown by Figure \ref{fig:parallel}, the optimal search strategy is quite intuitive---roughly speaking, the DM should stop searching and adopt alternative $i$ if and only if $x_i$ is relatively high compared with $x_j$ and the outside option of $0$, and she should stop searching and adopt the outside option when both $x_1$ and $x_2$ are relatively low. When $x_j$ is relatively low, the DM will continue to search on the two alternatives if and only if $x_i$ is near $0$, so as to make a clear distinction between alternative $i$ and the outside option. When both $x_1$ and $x_2$ are relatively high, the DM will continue to search if and only if $x_1$ and $x_2$ are close to each other, so as to to make a clear distinction between the two alternatives 1 and 2.

\section{Asymptotics}
\label{sc4}

\quad In this section, we study the free boundary of the solution near $x_1 = \ldots = x_d \rightarrow \infty$.
We provide a detailed analysis for the case with $d = 2$, and compare it with the case in which the DM can only search sequentially, learning one alternative at a time. We also provide lower and upper bounds for the general case with $d \ge 2$.

\subsection{Dimension of $d = 2$}
\quad In the case of $d=2,$ the PDE (\ref{eq:VIR}) specializes to
\begin{equation}
\label{eq:VIR2}
\min \left\{ -\frac{1}{2} \Delta u + c, u - \max\{x_1,x_2,0\} \right\} = 0.
\end{equation}

The PDE (\ref{eq:VIR2}) does not have an explicit solution for the case of $d=2$, so it is natural to ask about the properties of the solution, in particular those of free boundaries. There are three interesting regimes of asymptotic behavior:
\begin{enumerate}
\item
$x_1 \rightarrow 0$ and $x_2 \rightarrow -\infty$,
\item
$x_1 \rightarrow -\infty$ and $x_2 \rightarrow 0$,
\item
$x_1 = x_2 \rightarrow \infty$.
\end{enumerate}

 The cases 1 and 2 boil down to the search problem of one alternative, since the other alternative has large negative value and thus loses the competition to its counterpart.
A classical smooth-pasting technique shows that the distance of the free boundaries to $x$-axis (resp. $y$-axis) at $-\infty$ is $1/(4c)$, as illustrated in Figure \ref{fig:parallel}.
The case (3) is subtle, since the values of two products are close so there is a competitive search.
One interesting question is to determine the distance from the free boundary to the line $x_1= x_2$ at infinity.

 We start with the following change of coordinates:
$t =\frac{ x_1 + x_2}{\sqrt{2}}$ and $s = \frac{x_1 - x_2}{\sqrt{2}}$.
Consider the domain $t \ge 0$, and the PDE (\ref{eq:VIR2}) becomes
\begin{equation}
\label{eq:VI+}
 \displaystyle \min \left\{-\frac{1}{2} \Delta \tilde{u} + c, \tilde{u} - \tilde{g} \right\} = 0  \mbox{ for all } (t,s)\in \mathbb{R}^2
\end{equation}
where
\[\tilde{u}(t,s):=u\left(\frac{t+s}{\sqrt{2}},\frac{t-s}{\sqrt{2}}\right) \ \, \mbox{ and } \ \, \tilde{g}(t,s):=\max\left\{\frac{t + |s|}{\sqrt{2}},0\right\} .\]

We first prove a lower bound on the free boundary $\Gamma(u)$ for $t\geq0$ by the following lemma.
\begin{lemma}[Lower bound of the free boundary]
\label{lm:LB}
For $\theta>0$, let
\begin{equation}
\label{eq:gc}
\eta_\theta(t,s):=
\left\{ \begin{array}{lcl}
 \displaystyle \frac{t}{\sqrt{2}}+{\theta s^2} + \frac{1}{8\theta} & \mbox{for } |s| \le \frac{1}{2\sqrt{2}\theta} \\
 \displaystyle \frac{t }{\sqrt{2}}+\frac{|s|}{\sqrt{2}} & \mbox{for } |s| > \frac{1}{2\sqrt{2}\theta}.
\end{array},\right.
\end{equation}
which is a $\mathcal{C}^1$ function. Then we have,
$$
	\tilde{u}(t,s) \geq  \eta_c(t,s)\quad  \text{ in } \mathbb{R}^2.
$$
Moreover, for $t\geq0$, the free boundary $\Gamma(u)$ lies inside $\{|x_1-x_2|=\sqrt{2}|s|\geq\frac{1}{2c}\}$.
\end{lemma}

\begin{proof}
Note that $\eta_\theta$ is an approximation of $\tilde{g}$ for $t\geq 0$.
Moreover, it is not hard to check
when $\theta=c$,
\begin{equation*}
    \min \left\{- \frac{1}{2}\Delta \eta_c + c, \eta_c - \rho \right\}= 0 \quad \text{ for }(t,s)\in \mathbb{R}^2
\end{equation*}
where
\[\rho:=\frac{t + |s|}{\sqrt{2}}\leq \tilde{g}.\]
We know that
$\eta_c(\frac{x_1+x_2}{\sqrt{2}},\frac{x_1-x_2}{\sqrt{2}})$ is actually a subsolution to (\ref{eq:VIR2}) and the comparison principle yields $\tilde{u} \geq \eta_c$.
Observe that $\eta_c=\rho=\tilde{g}$ for $|s|\geq \frac{1}{2\sqrt{2}c}$ and $t\geq 0$. Therefore, in the half plane $t\geq 0$, the free boundary $\Gamma(u)$ lies inside $\{|s|\geq\frac{1}{2\sqrt{2}c}\}$.
\end{proof}

The result that $\Gamma(u)$ lies inside $\{|x_1-x_2|\geq\frac{1}{2c}\}$ can be viewed as a ``lower bound'' of the free boundary.

Now we turn to look for an ``upper'' bound of the free boundary. We need the following result.
\begin{lemma}
\label{lm:LU}
For $\epsilon\in (0,c]$, let
\begin{equation*}
{\varphi}_\epsilon(t,s) := \frac{1}{4c}h(\alpha t)+\eta_{c-\epsilon}(t,s)
\end{equation*}
where $h(t):=\max\{1-t,0\}^2$ and $\alpha:= 2\sqrt{c\epsilon}$.
Then we have for all $t\geq 0$,
\begin{equation*}
 \tilde{u}(t,s) \leq  {\varphi}_\epsilon(t,s).
\end{equation*}
\end{lemma}
\begin{proof}
It follows from (\ref{eq:boundmgle0}) that
\[ \tilde{g}(t,s)\leq \tilde{u}(t,s)\leq  \max\left\{\frac{t+s}{\sqrt{2}},0\right\}+\max\left\{\frac{t-s}{\sqrt{2}},0\right\}+\frac{1}{4c}.\]
We get
\begin{equation}
\label{cond:1}
{\frac{|s|}{\sqrt{2}} = } \tilde{g}(0,s)\leq \tilde{u}(0,s)\leq \frac{|s|}{\sqrt{2}}+\frac{1}{4c}.
\end{equation}

 Now we want to compare ${\varphi}_\epsilon$ with $u$ in the half plane $t>0$. On the boundary of $t=0$,
\begin{align*}
{\varphi}_\epsilon(0,s)&=\frac{1}{4c}+\eta_{c-\epsilon}(0,s)\\
&\geq \frac{1}{4c}+\tilde{g}(0,s) &\text{(by definition of $\eta_{c-\epsilon}$) }\\
&\geq \tilde{u}(0,s) &\text{(by (\ref{cond:1}))}.
\end{align*}
Also it is not hard to check that ${\varphi}_\epsilon\in\mathcal{C}^1$ and
${\varphi}_\epsilon(t,s)\geq \tilde{g}(t,s)$ for all $t\geq 0,s\in\mathbb{R}$.
Moreover when $|s|\leq\frac{1}{2\sqrt{2}(c-\epsilon)}$, we have
\begin{equation*}
\Delta {\varphi}_\epsilon=\frac{\alpha^2}{2c}+2(c-\epsilon)\leq 2c \hbox{ if
} \alpha\leq {2\sqrt{c\epsilon}}.
\end{equation*}
When $|s|\geq \frac{1}{2\sqrt{2}(c-\epsilon)}$, there is $\Delta {\varphi}_\epsilon\leq 2\epsilon\leq 2c$.
Finally note that both $\tilde{u}$ and ${\varphi}_\epsilon$ have linear growth at infinity.
The comparison principle (Lemma \ref{comparison} with $\Omega=\{t>0\}$) yields $\tilde{u}\leq {\varphi}_\epsilon$ for $t>0$.
\end{proof}

Based on Lemmas \ref{lm:LB} and \ref{lm:LU}, we can obtain the asymptotic behavior of solutions $u$ close to $x_1=x_2 \rightarrow +\infty$. To provide a quantitative description about the convergence of the free boundary of $u$ to the one of $\eta_c$ as $t=x_1+x_2\to \infty$, we define the distance function
\[d_{FB}(T):=\text{ distance between } \Gamma(u)|_{\frac{x_1+x_2}{\sqrt{2}}\geq T} \text{ and }\left\{|x_1-x_2|=\frac{1}{2c}\right\}. \]
Here $\{|x_1-x_2|=\frac{1}{2c}\}$ is the free boundary of $\eta_c$.
By symmetry of $\Gamma(u)$ with respect to the line of $x_1-x_2=0$, we only need to consider the situation when $x_1-x_2\geq 0$.

The following proposition characterizes the asymptotic behaviors of both the value function and the free boundary close to $x_1=x_2 \rightarrow +\infty$, with the proof provided in Appendix.
\begin{proposition} [Upper bound of the free boundary]
\label{prop:asymptotics}
For $(x_1,x_2)$ in the neighborhood of $x_1=x_2 \rightarrow +\infty$,
\begin{align*}
u(x_1,x_2) &\rightarrow \eta_c \left(\frac{x_1+x_2}{\sqrt{2}},\frac{x_1-x_2}{\sqrt{2}}\right)=
\left\{ \begin{array}{lcl}
 \displaystyle \frac{x_1+x_2}{2}+\frac{c|x_1-x_2|^2}{2}+\frac{1}{8c} & \mbox{for } |x_1 - x_2| \le \frac{1}{2c}  \\
 \displaystyle \frac{x_1+x_2}{2} + \frac{|x_1-x_2|}{2} & \mbox{for } |x_1 - x_2| > \frac{1}{2c}.
\end{array}\right.,\\
\Gamma(u) &\rightarrow \left\{|x_1 - x_2| = \frac{1}{2c}\right\}.
\end{align*}
Moreover, for all $T\geq \frac{1}{2c}$, then
\[d_{FB}(T)\leq \frac{1}{8\sqrt{2}c^3T^2}+O\left(\frac{1}{c^5T^4}\right).\]
\end{proposition}

As for the limit $\eta_c(x_1+x_2,x_1-x_2)$, the distance of the free boundary to the line of $x_1 = x_2$ is always $\frac{1}{2^{3/2}c}$.
Note that $\frac{1}{2^{3/2}c} > \frac{1}{4c}$ which is the distance of the free boundaries to $x$ or $y$-axis at $-\infty$.
This means that the search region is larger in case of competition.
In other words, people have larger tolerance for search if two products are as good as each other.

\subsection{Comparison with Sequential Search}

 \quad One could consider a different technology for information search, as the one considered in \cite{ke16}, where the DM searches costly and sequentially over multiple alternatives, learning only one alternative at a time. Let the sequential search cost be $c'$.

 Suppose $c' =c/2$. That is, it costs twice as much to search two alternatives in parallel as to search one alternative at a time. Note that in the sequential search case, the DM could replicate any parallel search strategy considered above by alternating infinitely fast between the two alternatives. Therefore, we have that the region in  $x_1$-$x_2$ space where it is optimal to continue to search (i.e. $\{u>g\}$) is larger for the case of sequential search compared with that for the case of parallel search. In other words, the contact set is further away from the origin for the case of sequential search.

Figure \ref{fig:comparison} illustrates the sequential and parallel search strategies for the case with $d=2$. The black solid lines represent the free boundaries for the case of parallel search, the same as Figure \ref{fig:parallel}; while the gray solid lines represent the free boundaries for the case of sequential search. The gray dashed line represents $x_1=x_2$. For the case of sequential search, when it is optimal for the DM to continue to search, the DM optimally searches alternative $i$ if and only if $x_i\geq x_j$ (\citealt{ke16}). The figure illustrates that the gray lines are further away from the origin than the black lines.
\begin{figure}[h]
\begin{center}
\includegraphics[width=0.5\textwidth]{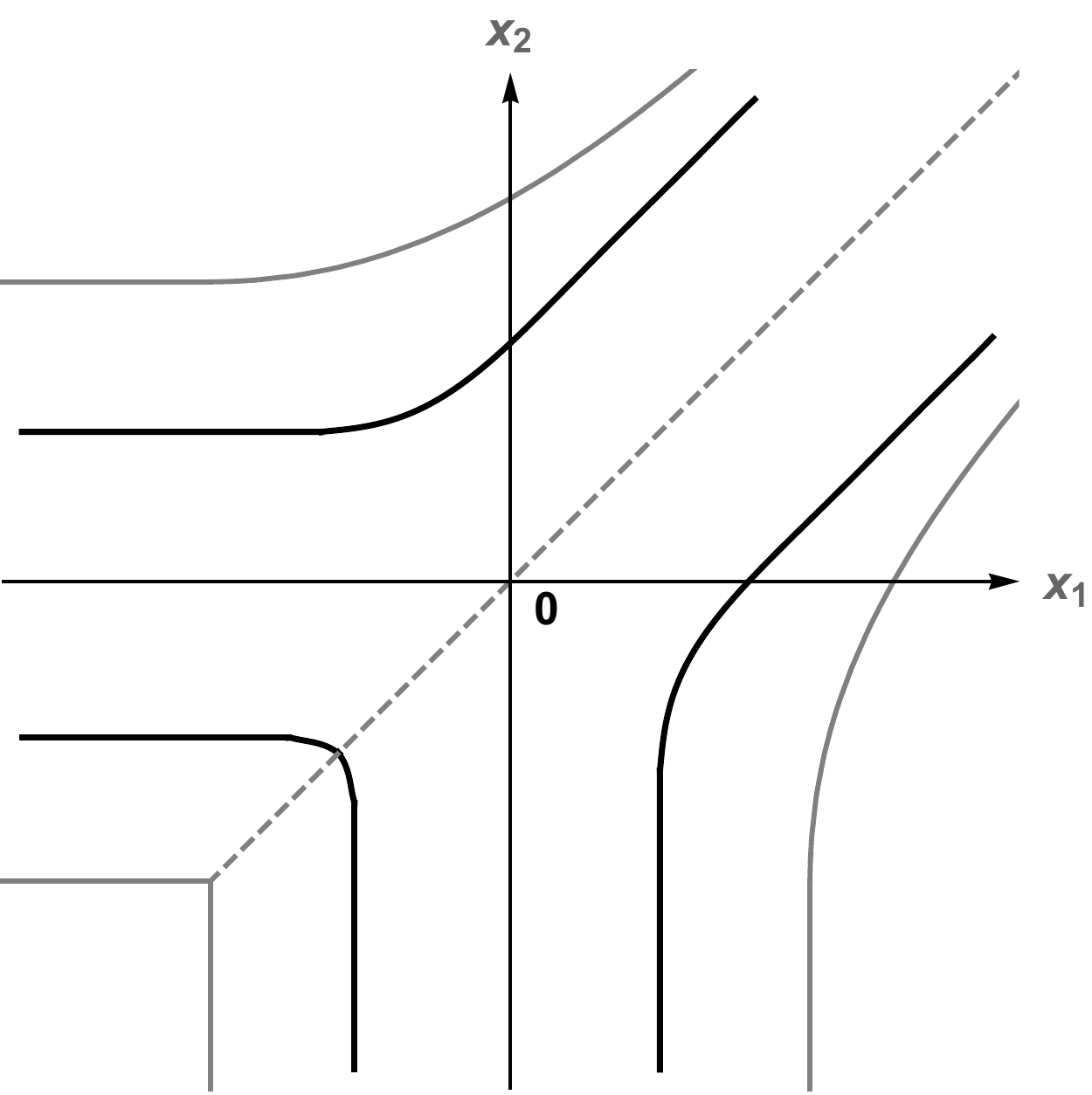}
\caption{Comparison of the optimal parallel search strategy under search cost $c$ with the optimal sequential search strategy under search cost $c'=c/2$.}
\label{fig:comparison}
\end{center}
\end{figure}

One could also wonder how the asymptotic behavior of the free boundary compares between sequential and parallel search. On can obtain that when the DM searches sequentially, the distance of the free boundary to the line $x_1 = x_2$ when $x_1=x_2 \rightarrow + \infty$ converges to $\frac{1}{4\sqrt{2}c'}= \frac{1}{2\sqrt{2}c}$ (\citealt{ke16}) which is the same as the distance in the case of parallel search. That is, the gray and black lines in Figure \ref{fig:comparison} will converge to $|x_1-x_2|= \frac{1}{2c}$ as $x_1$ and $x_2$ go to positive infinity.

It is also interesting to consider the case in which the DM has the option to either search only one alternative at cost $c'$ or search both alternatives at cost $c$ with economies of scale on the number of alternatives searched, that is, $c' \in (c/2,c).$ Although the full-scale analysis of this problem is beyond the scope of this paper, one could expect that in such a setting when it is optimal to continue to search for information, the DM will choose to search for information on both alternatives simultaneously when the expected valuations of the two alternatives are relatively close, and choose to search for information on only one alternative otherwise.

It is interesting, however, that one can obtain a general result in that setting that for the state close to $x_1 = x_2 \rightarrow \infty$ it is always optimal to choose the search technology where both alternatives are being searched simultaneously.

\begin{proposition} Consider a DM, who can search either in parallel at cost $c$ or sequentially at cost $c' \in (c/2, c)$. For $x_1$ and $x_2$ sufficiently high and close to each other, it is optimal for the DM to search in parallel.
\end{proposition}

Here we provide an intuitive sketch of proof for the proposition. A formal proof can be obtained by applying Lemma \ref{lem lowb} below and invoking the dynamic programming principle, and is omitted.

When $x_1$ and $x_2$ are high, the DM is most likely to choose one of the alternatives rather than the outside option, and just does not know which alternative to choose. The DM is then mostly concentrated on the difference $x_1 - x_2$ to see when this difference is high enough so that the DM makes a decision on which alternative to pick and stop the search process.
 As shown above, at the limit, when $|x_1-x_2|\geq \frac{1}{2c}$, the DM prefers to stop and choose one alternative than to continue to search either sequentially or in parallel. On the other hand, at the limit, when $|x_1-x_2|< \frac{1}{2c}$, the DM will choose to continue to search. By searching the two alternatives in parallel in an infinitesimal time $dt$, the DM pays a search cost of $cdt$ and gets an update on $x_1 - x_2$ as $dx_1-dx_2$, the variance of which is $2dt$; on the other hand, by searching one alternative (say, alternative 1) sequentially in an infinitesimal time $dt$, the DM pays a search cost of $c'dt$ and gets an update on $x_1 - x_2$ as $dx_1$, the variance of which is $dt$. Therefore, the parallel search yields variance per search cost $2/c$, which is greater than $1/c'$, the variance per search cost in the case of sequential search. To summarize, for $c' \in (c/2,c)$, it is less expensive to obtain a certain variation when searching two alternatives simultaneously, than just searching sequentially on one alternative. This implies that it is more cost-effective for the DM to search in parallel.

\begin{figure}[h]
\begin{center}
\includegraphics[width=0.5\textwidth]{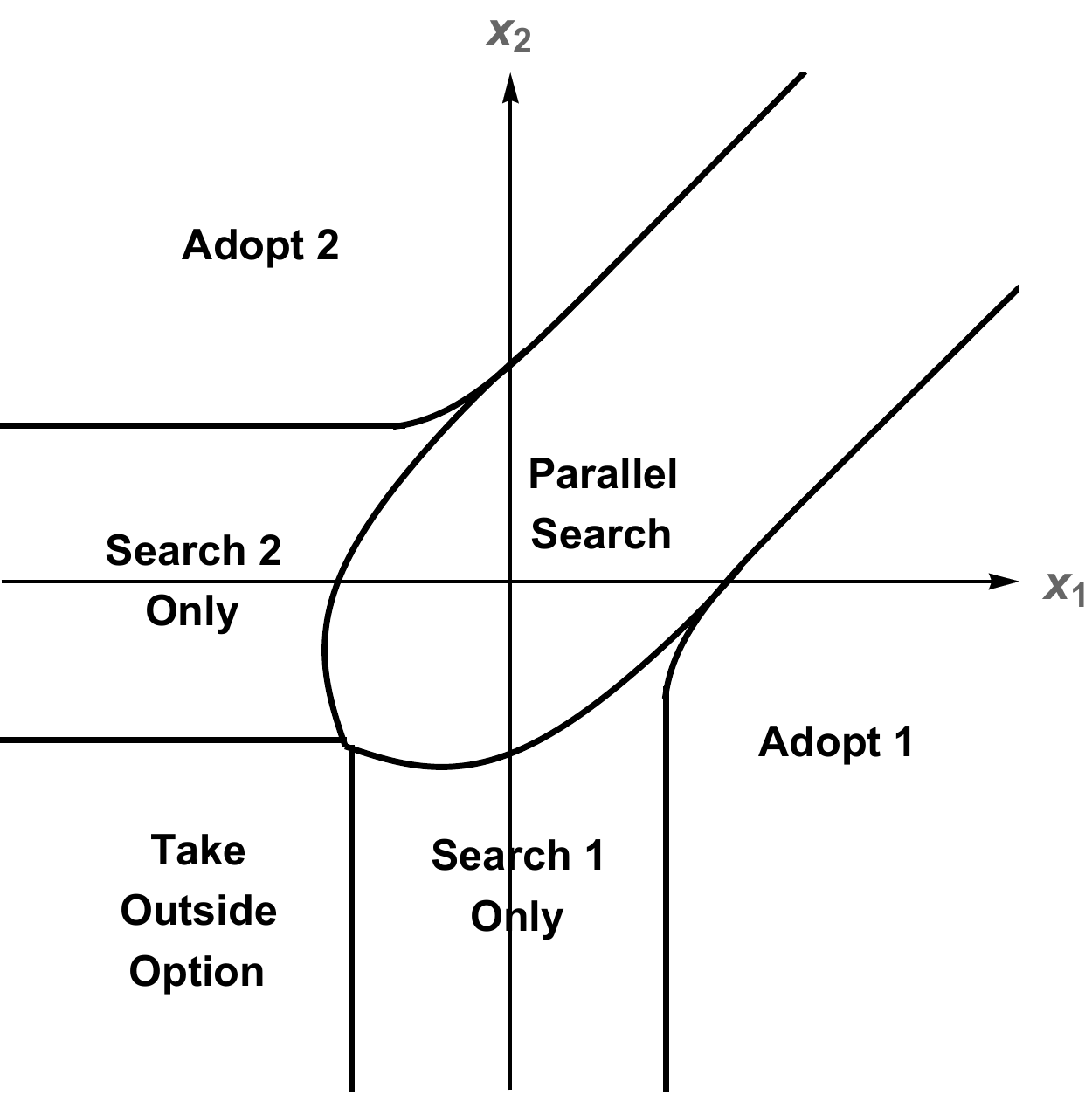}
\caption{The DM's optimal search strategy when he can search either sequentially or in parallel, with $c'=2c/3$.}
\label{fig:both}
\end{center}
\end{figure}
Figure \ref{fig:both} presents an example of the DM's optimal search strategy in this context of economics of scale of search costs, and illustrates that for $x_1$ and $x_2$ sufficiently high and close to each other, it is optimal to search the two alternatives in parallel.

\subsection{General dimension}

\quad Now we study the quantitative properties of the free boundary in the general dimension case. We provide an ``upper'' and ``lower'' bound of the free boundary.

 First for the upper bound, we will show that in the positive regime $x_i\geq 0$ for all $i$, the free boundary can not be too far away from the set $\{x_i=x_j \text{ for some }i\ne j\}$, and the distance grows at most linearly in the dimension $d$.


For any $\gamma>0$, define
\begin{equation}
    \label{set A}
    N(\gamma)=\left\{(x_1,...,x_d)| \, x_i\geq 0,\, |x_i-x_j|\geq {\gamma} \text{ for all }i\ne j\right\}.\end{equation}
The following proposition presents the first main result in this section, with the proof in the Appendix.

\begin{proposition}\label{thm ddim up}
Let $u_d$ be the solution to (\ref{eq:VIR}) in $\mathbb{R}^d$. There exists $\gamma>1$ independent of $d,c$ such that $\Gamma(u_d)$ lies inside the complement of $N(\frac{\gamma d}{c})$, i.e., $u_d(x)=g(x)$ for $x\in N(\frac{\gamma d}{c})$.
\end{proposition}

Note that in the case of parallel search here, for fixed $c,$ this result does not yield that the distance between the free boundary and $\{x_1=\ldots=x_d\}$ is bounded when $d \rightarrow \infty.$ We show that this is indeed the case when we investigate a ``lower bound'' of the free boundary in Proposition \ref{pr6}. In contrast, in the case of sequential search with cost $c'$, when $d \rightarrow \infty,$ that distance converges to $\frac{1}{\sqrt{2}c'}$ (see \cite{ke16}, p. 3591). However, if we set $c'=c/d$ for a fixed $c,$ then as $d \rightarrow \infty$, we would get $c' \rightarrow 0$, and correspondingly that distance for sequential search becomes unbounded. On the other hand, if we fix $c'$ and let $c=dc'$ grow linearly with $d$, then by Proposition \ref{thm ddim up}, we would have the distance between the free boundary and $\{x_1=\ldots=x_d\}$ for parallel search to be bounded.

 This can also be seen, by a similar argument used above, that as we can replicate parallel search by alternating among alternatives in sequential search, it must be that the ``search region'' (i.e., $\{u>g\}$) in the case of parallel search with cost $c$ is a subset of that in the case of sequential search with cost $c'=c/d$. As the distance between the free boundary and $\{x_1=\ldots=x_d\}$ is bounded for sequential search for a fixed $c'$, it must be that it is also bounded for parallel search for $c=dc'$. Note also that even though the free boundary is unbounded when $d \rightarrow \infty$ for fixed $c,$ the search process ends in finite time with probability one as the state moves, over time, away from $x_1 = ... = x_d.$ The question  of whether the distance for parallel search increases in $d$ for fixed $c'=c/d$ remains open.

Next we study the ``lower bound'' of the free boundary. Let us consider the following auxiliary problems: for each $d\geq 1$ consider
\begin{equation}
\label{eq:auxi}
\min \left\{ -\frac{1}{2} \Delta w_{d} + c, w_{d} - \rho \right\} = 0  \mbox{ in } \mathbb{R}^d,
\end{equation}
where $\rho=\max\{x_1,x_2,...,x_d\}$. The free boundary of $w_d$, ($\Gamma(w_d)$) is defined as the boundary of the set $\{w_d=\rho\}$.

 When $d=1,2$, by direct computation, we have that,
\[w_1=\psi_{c}, \mbox{ and } w_2=\eta_c,\]
where 
\begin{equation*}
\psi_c(x)=
\left\{ \begin{array}{lll}
 \displaystyle 0 & \mbox{for } x \le -\frac{1}{4\theta} \\
 \displaystyle c \left(x+\frac{1}{4c}\right)^2 & \mbox{for } x \in (-\frac{1}{4\theta}, \frac{1}{4\theta})\\
 \displaystyle x & \mbox{for } x \geq \frac{1}{4\theta},
\end{array}\right.
\end{equation*}
and $\eta_c$ is given by (\ref{eq:gc}).
Since $\rho\leq g$, $w_d$ is a subsolution to the original PDE (\ref{eq:VIR}) and by comparison $u(=u_d)\geq w_d$. We will show that $w_d$ provides the full information of the behavior of $u$ near $x_1=...=x_d\to\infty$.

\medskip

Let us introduce some notation. We write the positive $x_1,...,x_d$ directions as $e_1,...,e_d$ respectively and
\begin{equation}
    \label{notation t}
    \tau_d:=\frac{\sum_{i= 1}^d e_i}{\sqrt{d}},\, H_{\tau_d}:=\{v|\,v\cdot\tau_d=0\},\text{ and }t=\frac{\sum_{j=1}^{d} x_j}{\sqrt{d}}.
\end{equation}
The following lemma shows that we can reduce the study of $w_d$ to $H_{\tau_d}$, where the proof is provided in Appendix.

\begin{lemma}\label{lem proj}
The expression $w_d-\sum_{j=1}^d x_j/d$ is a constant function in the $\tau_d$ direction. The free boundary of $w_d$ is the surface of one infinitely long columnar with $\tau_d$ as its longitudinal axis.
\end{lemma}

 In the following lemma, we show that the free boundary of $w_d$ can be arbitrarily close to the one of $u$ if $\sum_{j=1}^{d} x_j$ is large. Since we are only interested in the region near $x_1=...=x_d$, let us define the following open neighborhood:
\[N_0(R):=
\left\{x\in \mathbb{R}^d\,| \, dist(x, \{s\tau_d, s\in\mathbb{R}\})\leq R\right\}.\]

\begin{lemma}\label{lem lowb}
For any $\epsilon\in(0,1)$ and $R\geq 1$, the distance between the free boundaries of $u$ and $w_d$ is bounded by $R\epsilon$ in the set
\[N_0(R)\cap \left\{ \sum_{j=1}^d x_j\geq \frac{d}{c}\sqrt{\frac{\gamma}{\epsilon}}\right\}\]
where $\gamma$ is a universal constant given in Proposition \ref{thm ddim up}.
\end{lemma}

We provide the proof to Lemma \ref{lem lowb} in the Appendix. Though more complicated, the idea of the proof follows from the one of Lemma \ref{lm:LU}.

We are interested in the most competitive region $x_1=...=x_d\to\infty$ where $d$ products are close. From Proposition \ref{thm ddim up}, we know that $\Gamma(u)$ can not be too far away from the axis $x_1=...=x_d$. Now we try to answer the question that how close this distance can be. By Lemma \ref{lem lowb}, we can identify $\Gamma(u)$ asymptotically with $\Gamma(w_d)$, and by Lemma \ref{lem proj}, we only need to study $\Gamma(w_d)\cap H_{\tau_d}$.

We make the following definition: for each $d\geq 1$, define $r_d$ to be the smallest number such that there exists $x\in H_{r_d}$ satisfying
\[|x|=r_d \text{ and } w_d(x)= \rho(x).\]
From the definition whenever $x\in H_{r_d}$ and $w_d(x)=\rho(x)=g(x)$, then $|x|\geq r_d$.
For example, when $d=1,2$, by the definition of $\psi_{c}$ and $\eta_c$, we have that $r_1= \frac{1}{4c}$ and $r_2=\frac{1}{2\sqrt{2}c}$.

Before the proposition, we need one technical lemma which compares $w_d$ and $w_{d'}$ for $d \ne d'.$

\begin{lemma}\label{cp wj}
For any $k>j\geq 1$, let $\{i_1,i_2,...,i_k\}$ be a permutation of $\{1,...,k\}$. Consider two solutions
$w_k(x_1,...,x_k)$ and $ w_j(x_{i_1},...,x_{i_j}).$
We can view $w_j$ as a function in $\mathbb{R}^k$ by trivial extension:
\[\tilde{w}_j(x_{i_1},...,x_{i_k}):=w_j(x_{i_1},...,x_{i_j}).\]
Then $w_k\geq \tilde{w}_j$ in $\mathbb{R}^k$.
\end{lemma}

The lemma is a direct result of the comparison principle.
With a slight abuse of notation, we still write $w_j$ instead of $\tilde{w}_j$.

\medskip

Now we prove the second main result of this section, which provides a lower bound on the free boundary.
\begin{proposition}\label{pr6}
Let $u=u_d$ be the solution to (\ref{eq:VIR}) in dimension $d$ and $r_d$ be given as the above.
In the half plane $\left\{t\geq \frac{1}{c}\sqrt{\frac{\gamma d}{\epsilon}}\right\}$, the distance from the free boundary of $u$ to the ray $\{s\tau_d, \,s\geq 0\}$ lies in the interval \[[\,r_d\,,\,r_d+\frac{\gamma d}{c}\epsilon\,],\]
where $\tau_d$ is defined in (\ref{notation t}), and $\gamma$ is a universal constant given in Proposition \ref{thm ddim up}.
Furthermore, for each $d\geq 3$,
\[\left(d-2-\frac{1}{d}\right)r_d^2\geq (d-2)r^2_{d-1}.\]
In particular, $r_d$ is increasing in $d.$ Furthermore, $r_d \rightarrow \infty$ as $d \rightarrow \infty$.
\end{proposition}
\begin{proof}
Due to Proposition \ref{thm ddim up}, $r_d\leq \frac{\gamma d}{c}$. We apply Lemma \ref{lem lowb} with $R=\frac{\gamma d}{c}$ and then the first part of the result follows from the definition of $r_d$.

 For the second part, take $k\geq 3$ and $x\in H_{\tau_k}$. Without loss of generality we assume
\[\rho_k(x):=\rho(x_1,...,x_k)=\max\{x_1,...,x_k\}=x_1>0.\]
The inequality holds due to $x\in H_{\tau_k}$. Suppose $w_k(x)=x_1=\rho_k(x)$. Take any $k-2$ different numbers \[\{{i_2},{i_3},...,{i_{k-1}}\}\subset \{2,...,k\}.\]
If\[x_{1}^2+x^2_{i_2}+...+x^2_{i_{k-1}}<r^2_{k-1},\]
by Lemma \ref{cp wj} it follows that
\[w_k(x_1,...,x_k)\geq w_{k-1}(x_{1},x_{i_2},...,x_{i_{k-1}})>\max\{x_{1},x_{i_2},...,x_{i_{k-1}}\}=x_1=\rho_k(x),\]
which cannot happen due to our assumption $w_k=\rho_k$ at $x$. Thus we must have
\[x_{1}^2+x^2_{i_2}+...+x^2_{i_{k-1}}\geq r_{k-1}^2.\]
We can vary the subscripts and add up all the inequalities with respect to different combinations of $\{i_2,...,i_{k-1}\}$. It ends up with
\begin{equation}
    \label{eqn 0.14 1}
    (k-2)x_1^2+(k-3)(\sum_{j=2}^k\, x_j^2)\geq (k-2)r^2_{k-1}.
\end{equation}

 Due to the facts that $\sum_{j=1}^k x_j=0$ and $x_1=\max\{x_1,...,x_k\}$, we can show
\[x_1^2\leq \frac{k-1}{k}|x|^2,\]
and equality can be obtained when,
\[x_1=\sqrt{\frac{k-1}{k}}|x|,\, x_2=x_3=...=x_k=-\frac{1}{\sqrt{{k(k-1)}}}|x|^2.\]
Therefore (\ref{eqn 0.14 1}) leads to
\[\left(k-3+\frac{k-1}{k}\right)|x|^2\geq (k-2)r^2_{k-1}.\]
According to the assumption $w_k=\rho_k$ and the definition of $r_k$,
\[r_k^2\geq|x|^2\geq  \frac{(k-2)}{(k-2-\frac{1}{k})}r^2_{k-1}.\]

To prove that $r_d \rightarrow \infty$ as $d \rightarrow \infty,$ suppose by contradiction that $r_d$ is bounded as $d \rightarrow \infty$.
Then for each $d \ge 2$, there exists $(x^d_1, \ldots , x^d_d) \in \{ u = g\}$ such that $\sup_{1 \le i \le d}x_i^d \le K$ for some
$K$ independent of $d$.
It is well known that for $Z_1, \ldots, Z_d$ i.i.d. $\mathcal{N}(0,1)$,
$\mathbb{E}(\max_{1 \le i \le d} Z_i) \sim \sqrt{2 \log d}$ as $d \rightarrow \infty$.
This implies that given any $c > 0$,
\begin{equation*}
\mathbb{E}\left(\max\left(B_1^{x^d_1}(1), \ldots, B_d^{x^d_d}(1), 0\right) \right)> \max(x^d_1, \ldots, x^d_d, 0) + c.
\end{equation*}
for $d$ sufficiently large.
Therefore,
\begin{equation*}
u(x^d_1, \ldots, x^d_d) \ge \mathbb{E} \left(\max\left(B_1^{x^d_1}(1), \ldots, B_d^{x^d_d}(1), 0\right)\right) - c > g(x^d_1, \ldots, x^d_d).
\end{equation*}
This contradicts the fact that $(x^d_1, \ldots , x^d_d) \in \{ u = g\}$.
This concludes the proof.
\end{proof}


\appendix
\setcounter{lemma}{0}
\setcounter{definition}{0}
\renewcommand{\thelemma}{A\arabic{lemma}}
\renewcommand{\thedefinition}{A\arabic{definition}}

\section*{Definition of Viscosity Solutions}

\begin{definition}
\label{def3}
Let $u$ be a continuous function and $x^0\in\mathbb{R}^d$.
\begin{enumerate}
\item
We say that $-\mathcal{L} u  \le f$ at $x^0$ in the viscosity sense if for any $\varphi\in\mathcal{C}^2$ which touches $u$ at $x^0$ from above, we have $-(\mathcal{L} \varphi)(x^0)  \le f(x^0)$.
We call $u$ a subsolution to (\ref{eq:PDE}) if $-\mathcal{L} u  \le f$ in the viscosity sense at all points where $u - g>0$.

\item
We say that $-\mathcal{L} u  \ge f$ at $x^0$ in the viscosity sense if for any $\varphi\in\mathcal{C}^2$ which touches $u$ at $x^0$ from below, we have $-(\mathcal{L} \varphi)(x^0)  \ge f(x^0)$.
We call $u$ a supersolution to (\ref{eq:PDE}) if $u - g\geq 0$ and $-\mathcal{L} u  \geq f$ in the viscosity sense in $\mathbb{R}^d$.

\item
We call $u$ a viscosity solution to (\ref{eq:PDE}) if and only if $u$ is both a subsolution and a supersolution to (\ref{eq:PDE}).
\end{enumerate}

\end{definition}

\section*{Lemma for General Case}

\begin{lemma}
\label{lem:VI+bd0}
Let $u$ be the value function defined by (\ref{eq:valuegeneral}), with $\mathcal{T}: = \{\tau \mbox{ is a stopping time}: \mathbb{E} \tau < \infty\}$.
Then $u$ is a viscosity solution to
\begin{equation}
\label{eq:VIR0}
 \min \left\{-\mathcal{L} u - f, u - g \right\} = 0.
\end{equation}
Moreover, if there exist $K_1, K_2 > 0$ such that
\begin{equation*}
\sum_{i = 1}^d \sup_{x \in \mathbb{R}^d}|b_i(x)| < K_1 \quad \mbox{and} \quad  \sup_{i,j} \sup_{x \in \mathbb{R}^d}|\sigma_{ij}(x)| < K_2,
\end{equation*}
and there exist $c > K_1$ and $a > 0$ such that
\begin{equation*}
\max_{x \in \mathbb{R}^d} f(x) \le -c \quad \mbox{and} \quad g(x) \le  a \sum_{i = 1}^d |x_i| \,\, \mbox{for all } x \in \mathbb{R}^d,
\end{equation*}
then we have for some $C > 0$,
\begin{equation}
\label{eq:boundmgle0}
g \le u \le a \sum_{i = 1}^d |x_i| + C.
\end{equation}
\end{lemma}
\begin{proof}
The fact that $u$ is a viscosity solution to (\ref{eq:VIR0}) follows from the dynamic programming principle (\ref{eq:PDE}).
By taking $\tau = 0$, we get $u \ge g(x_1, \ldots, x_d)$.
Moreover,
\begin{align*}
g(X^{x_1}_1(\tau), \ldots,X^{x_d}_d(\tau))  & \le a \sum_{i = 1}^d |X_i^{x_i}(\tau)| \\
& \le a \left[ \sum_{i = 1}^d |x_i| +  \sum_{i = 1}^d 
\left| \int_0^{\tau} b_i(X^x(s)) ds + \sum_{j = 1}^d \int_0^{\tau} \sigma_{ij}(X^x(s)) dB_j(s) \right| \right] \\
& \le a \left[ \sum_{i = 1}^d |x_i| +  \sum_{i = 1}^d \int_0^{\tau} |b_i(X^x(s))| ds  + \sum_{i = 1}^d \sum_{j = 1}^d \left|\int_0^\tau \sigma_{ij}(X^x(s)) dB_j(s) \right| \right].
\end{align*}
Note that
\begin{equation*}
\mathbb{E} \left( \sum_{i = 1}^d \int_0^{\tau} |b_i(X^x(s))| ds \right) \le K_1 \mathbb{E}\tau,
\end{equation*}
and there exists $L >0$ such that for any $1 \le i,j \le n$,
\begin{align*}
\mathbb{E} \left|\int_0^\tau \sigma_{ij}(X^x(s)) dB_j(s) \right| &\stackrel{(*)}{\le} L \mathbb{E} \left[ \left( \int_0^{\tau} \sigma_{ij}^2(X^x(s)) ds \right)^{\frac{1}{2}} \right] \\
& \le LK_2 \sqrt{\mathbb{E} \tau}.
\end{align*}
where the inequality $(*)$ is due to the {\em Burkholder-Davis-Gundy inequality} (see  \citealt{RY}, Chapter IV).
Consequently,
\begin{align*}
u  & \le a \left[\sum_{i = 1}^d |x_i| +  \sup_{\tau \in \mathcal{T}} \left\{ (K_1 - c)\mathbb{E}\tau + L d^2 K_2 \sqrt{\mathbb{E} \tau} \right\} \right] \\
& \le a \sum_{i = 1}^d |x_i| +  \frac{a L^2 d^4 K_2^2}{4(c- K_1)},
\end{align*}
which yields (\ref{eq:boundmgle0}).
\end{proof}

\section*{Proof of Lemma \ref{comparison}:}
\vspace{-.2in}
\begin{proof}
Fix $r\geq 1$.
For any $R>r+1$ and $\epsilon\in (0,1/d)$, let
\[u_2^\epsilon:=(1-d\epsilon)u_2+c\epsilon \left(\sum_{i=1}^dx_i^2+\frac{d^3}{4c^2}\right).\]
We claim that $u_2^\epsilon$ is a supersolution to 
\begin{equation}
\label{eq:VI}
\left\{ \begin{array}{lcl}
 \displaystyle \min \left\{-\frac{1}{2} \Delta u + c, u -g \right\} = 0  \mbox{ for }  x\in\mathcal{B}_R, \\
 \displaystyle u =g  \mbox{ for } x\in\partial \mathcal{B}_R.
\end{array}\right.
\end{equation}
Since $\Delta u_2\leq 2c$, we have
\[\Delta u_2^\epsilon=(1-d\epsilon)\Delta u_2+2cd\epsilon\leq 2c.\]
Also because $u_2\geq g$, we get
\begin{align*}
    u_2^\epsilon&=(1-d\epsilon)u_2+\epsilon\sum_{i=1}^d \left(cx^2+\frac{d^2}{4c}\right)\\
    &\geq (1-d\epsilon)\max\{x_1,...,x_d,0\}+d\epsilon \sum_{i=1}^d|x_i|\\
    &\geq \max\{x_1,...,x_d,0\}=g.
\end{align*}

 Next for any small $\epsilon>0$, if we pick $R$ large enough (depending on $h$ and $\epsilon$) and then by the condition (\ref{cond: less2g}),
\[u_1\leq c\epsilon R^2\leq u_2^\epsilon  \text{ on $\partial\mathcal{B}_R\cup \partial\Omega$}. \]
By comparison, $u_2^\epsilon\geq u_1$ in $\mathcal{B}_R\cap\Omega$ and in particular in $\mathcal{B}_r\cap\Omega$.
Consequently,
\[(1-d\epsilon)u_2+c\epsilon \left(r^2+\frac{d^3}{4c^2}\right)\geq u_1.\]
Since we can choose $\epsilon$ to be arbitrarily small and then $r$ to be large, we conclude that
$u_1\leq u_2$ in $\Omega$.
\end{proof}

\section*{Proof of Proposition \ref{prop:asymptotics}:}
\vspace{-.2in}
\begin{proof}
Consider the line $x_1+x_2=\sqrt{2}t$ with fixed $t\geq 1$. By Lemma \ref{lm:LU},
\[               u\left(\frac{t+s}{\sqrt{2}},\frac{t-s}{\sqrt{2}}\right)=\tilde{u}(t,s)\leq {\varphi}_\epsilon(t,s).\]
By definition, using the notation in Lemma \ref{lm:LU}, when
\begin{equation}
    \label{t epsilon}
    \alpha t\geq 1  \text{ i.e. } \epsilon\geq 1/({4ct^2}),
\end{equation}
we have
$\tilde{u}\leq {\varphi}_\epsilon=\eta_{c-\epsilon}.$ This, combining with the fact that $\tilde{u}\geq \eta_{c+\epsilon}$, implies,
\begin{align*}
    \tilde{u}(t,s)=\,&u(x_1,x_2)\geq \max\{x_1,x_2,0\}=g(x_1,x_2)& \text{ if $|x_1-x_2|\geq \frac{1}{2c}$};\\
 \tilde{u}(t,s)=\,&u(x_1,x_2)\geq \frac{c|x_1-x_2|^2}{2}+\frac{1}{8c}+\frac{x_1+x_2}{2}>g(x_1,x_2) & \text{ if $|x_1-x_2|< \frac{1}{2c}$};\\
   \tilde{u}(t,s)=\,&u(x_1,x_2)\leq g(x_1,x_2)& \text{ if $|x_1-x_2|\geq \frac{1}{2(c-\epsilon)}$}.
\end{align*}
Thus,
\begin{align*}
&   u(x_1,x_2)> g(x_1,x_2)   \text{ if $|x_1-x_2|< \frac{1}{2c}$},\\
&    u(x_1,x_2)= g(x_1,x_2)  \text{ if $|x_1-x_2|\geq \frac{1}{2(c-\epsilon)}$}.
\end{align*}
We see that free boundary is between $|x_1-x_2|\in( \frac{1}{2c}, \frac{1}{2(c-\epsilon)})$ once $t=\frac{x_1+x_2}{\sqrt{2}}$ satisfying (\ref{t epsilon}).
Now take $\epsilon=\frac{1}{4ct^2}$ and to have $\epsilon<c$, we require $t\geq \frac{1}{2c}$.
Finally, we conclude that,
\begin{align*}
    d_{FB}(t)&\leq \left(\frac{1}{2(c-\epsilon)}-\frac{1}{2c}\right)/\sqrt{2}\\
    &= \frac{\epsilon}{2\sqrt{2}c(c-\epsilon)}\\
    &=\frac{\epsilon}{2\sqrt{2}c^2}+O\left(\frac{\epsilon^2}{c^3}\right)=\frac{1}{8\sqrt{2}c^3t^2}+O\left(\frac{1}{c^5t^4}\right).
\end{align*}
\end{proof}

\section*{Proof of Proposition \ref{thm ddim up}:}
\vspace{-.2in}
\begin{proof}
We first prove the following technical lemma.
\begin{lemma}\label{est Id}
There exists a universal constant $C$ such that for all $d\geq 2$
\[\int_0^1 \left|\frac{d}{dR}e^{-1/(1-R^2)}\right|R^{d-1}dR\leq C(d-1)\int_0^1 e^{-1/(1-R^2)}R^{d-1}dR.\]
\end{lemma}
\begin{proof}
Denote
\begin{equation}
    \label{J_d}
    J_d:=\int_0^1 e^{-1/(1-R^2)}R^{d}dR.
\end{equation}
Integration by parts gives
\begin{align*}
    \int_0^1 \left|\frac{d}{dR}e^{-1/(1-R^2)}\right|R^{d-1}dR&=\int_0^1 \left(-\frac{d}{dR}e^{-1/(1-R^2)}\right)R^{d-1}dR\\
    &=(d-1)\int_0^1 e^{-1/(1-R^2)}R^{d-2}dR=(d-1)J_{d-2}.
\end{align*}
By the Cauchy-Schwarz inequality, we have
\[J^2_{d-2}\leq J_{d-1}J_{d-3}.\]
Thus,
\begin{multline*}
    \int_0^1 \left|\frac{d}{dR}e^{-1/(1-R^2)}\right|R^{d-1}dR / \left((d-1)\int_0^1 e^{-1/(1-R^2)}R^{d-1}dR\right) \\
    =J_{d-2}/J_{d-1}\leq J_{d-3}/J_{d-2}\leq...\leq J_2/J_1=C.
\end{multline*}
\end{proof}

Now, we prove the main proposition. From previous arguments, we know that $g$ is a subsolution and $u\geq g$.
We are going to construct a supersolution through $g$ and it leads to an estimate of $\Gamma(u)$ from above.

 Consider a symmetric modifier \[\varphi(x)=\mu(|x|)/I_d\] such that
\begin{equation*}
\mu(R)=
\left\{ \begin{array}{lcl}
 \displaystyle e^{-1/(1-R^2)} &\text{ if }R\leq 1, \\
 \displaystyle 0 &\text{ if }R> 1.
\end{array}\right.
\end{equation*}
The numerical constant $ I_{d}$ ensures normalization, i.e.,
\[I_d=\int_{\mathbb{R}^d}\varphi(x)dx=A_dJ_{d-1}\]
where $A_d$ is the surface area of a unit $d$-dimensional ball and $J_{d-1}$ is given in (\ref{J_d}).

 Set $\varphi_r:=r^{d}\varphi(rx)$ and then 
\begin{equation}
\label{est varphi}
\begin{aligned}
    &supp\{\varphi_r\}\subset \{|x|\leq 1/r\},\\
    &\int_{\mathbb{R}^d}| \nabla\varphi_r|dx=r\int_{\mathcal{B}_1}| \nabla\varphi|dx\\
     &=\frac{r}{I_d}\iint_{\mathcal{B}_1}|\nabla\mu(R)|R^{d-1}dRd\omega=\frac{rA_d}{I_d} \int_0^1|\mu'(R)|R^{d-1}dR,
\end{aligned}
\end{equation}
where $\nabla$ is the gradient operator.
According to Lemma \ref{est Id},
\[\int_{\mathbb{R}^d}| \nabla\varphi_r|dx\leq C(d-1)rA_dJ_{d-1}/I_d=C(d-1)r.\]

 We claim that
\[\Phi_r:=\varphi_r *g=\int_{\mathbb{R}^d}\varphi_r(x-y)g(y)dy\] is a supersolution for some $r$ small enough. Let us check the following two conditions,
\begin{equation*}
\Delta \Phi_r\leq 2c,  \text{and} \Phi_r \geq g.
\end{equation*}

 Since (by symmetry) $\varphi_r* x_i=x_i$ and $g=\max\{x_1,...,x_d,0\}$, we have
\[\Phi_r=\varphi_r*g\geq  g.\]
Next we compute
\begin{align*}
    \left|\Delta \Phi_r \right|&= |\nabla_x \int_{\mathbb{R}^d}(\nabla\varphi_r)(x-y)g(y)dy|\\
    &= \left|\nabla_x \int_{\mathbb{R}^d}(\nabla\varphi_r)(y)g(x-y)dy\right|\\
    &\leq  \int_{\mathbb{R}^d}|\nabla\varphi_r|(y)|\nabla g|(x-y)dy.\end{align*}
By the fact $|\nabla g|\leq 1$ and (\ref{est varphi}), we obtain
\begin{align*}
|\Delta\Phi_r|\leq C(d-1)r.
\end{align*}
Thus for some universal $\gamma>1$, we have
$ \left|\Delta \Phi_r \right|\leq 2c$ if $r\leq c/ (\gamma d)$. In all we conclude that with this choice of $r$, $\Phi_r$ is a supersolution and $u\leq \Phi_r$.

 Fix any $x^0\in N(\gamma d/c)$. By definition,
\[\text{$g(x)=x_k$ for some $k$ for all $x=(x_1,...,x_d)\in \mathcal{B}_{\gamma d/c}(x^0)$}\]
and therefore $\Phi_r=g*\phi_r=x_k$. Hence in $ N(\gamma d/c)$, we have $u\leq \Phi_r=g.$
Since $u\geq g$, we conclude that
$u=g \text{ for }x\in N(\gamma d/c)$.
\end{proof}

\section*{Proof of Lemma \ref{lem proj}:}
\vspace{-.2in}
\begin{proof}
Let us sketch the proof below. We are going to use the following cylindrical coordinates: for each $x\in \mathbb{R}^d$, write
\[x=t \tau_d+\sum_{j=1}^{d} s_j e_j,\]
where $\sum_{j=1}^{d} s_j e_j\in H_{\tau_d}$. Then $\omega:=w_d-\frac{t}{\sqrt{d}}$ solves
\begin{equation}
    \label{eqn shift}
    \min \left\{ -\frac{1}{2} \Delta \omega + c, \omega - \left(\rho- \frac{t}{\sqrt{d}}\right)\right\} = 0  \mbox{ in } \mathbb{R}^d.
\end{equation}
Notice that shifts in the $\tau_d$ direction preserve the value of $(\rho- \frac{t}{\sqrt{d}})$. Therefore by uniqueness of solutions to (\ref{eqn shift}), the shifts also preserve $\omega$ i.e. $\omega(x)=\omega(x+s\tau_d)$ for all $s\in\mathbb{R}$.

 Now we consider the free boundary property of $w_d$. Again, for any $x\in \mathbb{R}^d$, write
$x=t\tau_k+y\text{ with }y\in H_{\tau_d}$ and $t(x)=\sum_{j=1}^d x_j/\sqrt{d}$. From the above, $w_d(y)=\rho(y)$ if and only if
\begin{align*}
    w_d(y)+\frac{t}{\sqrt{d}}= \rho(y)+\frac{t}{\sqrt{d}},
\end{align*}
if and only if
\[ w(x)=\max\{y_1,y_2,...,y_d\}+\frac{t}{\sqrt{d}}=\max\{x_1,x_2,...,x_d\}=\rho(x).\]
We used $x=y+t\tau_d$ in the second equality.
Therefore $\Gamma(w_d)$ equals $\{\Gamma(w_d)\cap H_{\tau_d}\}\times \mathbb{R}\tau_d$.
\end{proof}

\section*{Proof of Lemma \ref{lem lowb}:}
\vspace{-.2in}
\begin{proof}
First, we want to give an upper bound of $u-g$ on $H_{\tau_d}$. From the proof of Proposition \ref{thm ddim up}, $\Phi_{1/r}=\varphi_{1/r}*g\geq u$ where $\varphi_{1/r}$ is a modifier supported in $\mathcal{B}_{r}$ with $r=\frac{\gamma d}{c}$.
Because $|\nabla g|\leq 1$ and $\varphi_{1/r}*g (x)$ can be viewed as a weighted average of $g$ in $\mathcal{B}_r(x)$, we have
\[ |\Phi_{1/r}-g|\leq r.\]
In all, we find for $x\in H_{\tau_d}$
\begin{equation}\label{ddim bd}
    u-g\leq \Phi_{1/r}-g\leq r.
\end{equation}

Second, let us construct a supersolution to (\ref{eq:VIR}). For $\epsilon \ll \min\{c,1/d\}$, set \[w^\epsilon_d:= \frac{1}{1-\epsilon}w_d((1-\epsilon)x),\]
which then solves
\begin{equation*}
\min \left( -\frac{1}{2} \Delta w^\epsilon_{d} + (1-\epsilon)c, w^\epsilon_{d} -\rho \right) = 0.
\end{equation*}
Next define a $\mathcal{C}^1$ function
\begin{equation*}
{\varphi}_d^\epsilon(x) := r h(\alpha\, x\cdot\tau_d)+w_d^\epsilon(x),
\end{equation*}
where $h(t)=(\max\{1-t,0\})^2$ and $ \alpha=\alpha(\epsilon)$ are to be determined.

In the third step, we want to show that ${\varphi}_d^\epsilon$ is indeed a supersolution in the half hyperplane $\mathcal{D}:=\left\{\frac{\sum_{j=1}^d x_j}{\sqrt{d}}=:t>0\right\}$.
Since $\rho=g$ in $\mathcal{D}$,
we have ${\varphi}_d^\epsilon\geq g$ in $\mathcal{D}$.
On the boundary $\partial \mathcal{D}=H_{\tau_d}$, it follows from (\ref{ddim bd}) that
\[{\varphi}_d^\epsilon=rh(0)+w_d^\epsilon\geq r+g\geq u.\]
Also by direct computation,
\begin{equation*}
\Delta {\varphi}_d^\epsilon=\Delta (r h)+\Delta w_d^\epsilon \leq 2r\alpha^2+2(1-\epsilon)c.
\end{equation*}
To make ${\varphi}_d^\epsilon$ a subsolution, we only need $r\alpha^2\leq c\epsilon$ which is equivalent to
$\alpha\leq c\sqrt{\frac{\epsilon}{\gamma d}}$.
Finally we can conclude that by comparison, ${\varphi}_d^\epsilon\geq u$ in $\mathcal{D}$.

 When $t\geq \frac{1}{\alpha}=\frac{1}{c}\sqrt{\frac{\gamma d}{\epsilon}}$, we have ${\varphi}_d^\epsilon=w_d^\epsilon\geq u$. Hence we know $w_d^\epsilon\geq u\geq w_d$. Since
\[w_d^\epsilon(x)=\frac{1}{1-\epsilon}w_d((1-\epsilon)x),\] then $w_d(x)=g(x)$ implies $w_d^\epsilon(x)=g(x)$. Therefore the free boundary of $u$ ($\Gamma(u)$) lies between $\Gamma(w_d)$ and $\Gamma(w_d^\epsilon)$ when $t$ is large.
By Lemma \ref{lem proj}, it is sufficient to compare $\Gamma(w_d)$ and $\Gamma(w_d^\epsilon)$ on $H_{\tau_d}$.

We consider a $R$-neighbourhood of the origin in $H_{\tau_d}$ (seeing from Proposition \ref{thm ddim up}, we may pick $R=\frac{\gamma d}{c}$). Again by definition of $w_d^\epsilon$, inside $\mathcal{B}_{R}(t\tau_d)\cap H_{\tau_d}$, the distance between $\Gamma(w_d)$ and $\Gamma(w_d^\epsilon)$ is bounded by
$R\epsilon.$ 
We conclude that the distance between $\Gamma(u)$ and $\Gamma(w_d)$ is bounded by $R\epsilon$ in the set
$N_0(R)\cap \{ t\geq \frac{1}{c}\sqrt{\frac{\gamma d}{\epsilon}}\}$.
\end{proof}

\newpage
\bibliographystyle{plainnat}
\bibliography{unique}

\begin{thebibliography}{29}
\providecommand{\natexlab}[1]{#1}
\providecommand{\url}[1]{\texttt{#1}}
\expandafter\ifx\csname urlstyle\endcsname\relax
  \providecommand{\doi}[1]{doi: #1}\else
  \providecommand{\doi}{doi: \begingroup \urlstyle{rm}\Url}\fi

\bibitem[Assing et~al.(2014)Assing, Jacka, and Ocejo]{AJO14}
Sigurd Assing, Saul Jacka, and Adriana Ocejo.
\newblock Monotonicity of the value function for a two-dimensional optimal
  stopping problem.
\newblock \emph{The Annals of Applied Probability}, 24\penalty0 (4):\penalty0
  1554--1584, 2014.

\bibitem[Bergemann and V\"{a}lim\"{a}ki(1996)]{ber96}
Dirk Bergemann and Juuso V\"{a}lim\"{a}ki.
\newblock Learning and strategic pricing.
\newblock \emph{Econometrica}, 64\penalty0 (5):\penalty0 1125--1149, 1996.

\bibitem[Branco et~al.(2012)Branco, Sun, and Villas-Boas]{bra12}
Fernando Branco, Monic Sun, and J.~Miguel Villas-Boas.
\newblock Optimal search for product information.
\newblock \emph{Management Science}, 58\penalty0 (11):\penalty0 2037--2056,
  2012.

\bibitem[Broadie and Detemple(1997)]{bro97}
M.~Broadie and J.~Detemple.
\newblock The valuation of {A}merican options on multiple assets.
\newblock \emph{Mathematical Finance}, 7\penalty0 (3):\penalty0 241--286, 1997.

\bibitem[Caffarelli(1998)]{caffarelli1998obstacle}
Luis Caffarelli.
\newblock The obstacle problem revisited.
\newblock \emph{Journal of Fourier Analysis and Applications}, 4\penalty0
  (4-5):\penalty0 383--402, 1998.

\bibitem[Che and Mierendorff(2019)]{che16}
Yeon-Koo Che and Konrad Mierendorff.
\newblock Optimal sequential decision with limited attention.
\newblock \emph{American Economic Review}, 108\penalty0 (8):\penalty0
  2993--3029, 2019.

\bibitem[Crandall and Lions(1983)]{crandall1983viscosity}
Michael Crandall and Pierre-Louis Lions.
\newblock Viscosity solutions of hamilton-jacobi equations.
\newblock \emph{Transactions of the American mathematical society},
  277\penalty0 (1):\penalty0 1--42, 1983.

\bibitem[Crandall et~al.(1992)Crandall, Ishii, and Lions]{CIL}
Michael Crandall, Hitoshi Ishii, and Pierre-Louis Lions.
\newblock User's guide to viscosity solutions of second order partial
  differential equations.
\newblock \emph{Bulletin of the American Mathematical Society}, 27\penalty0
  (1):\penalty0 1--67, 1992.

\bibitem[Frehse(1972)]{frehse1972regularity}
Jens Frehse.
\newblock On the regularity of the solution of a second order variational
  inequality.
\newblock \emph{Boll. Un. Mat. Ital.(4)}, 6:\penalty0 312--315, 1972.

\bibitem[Fudenberg et~al.(2018)Fudenberg, Strack, and Strzalecki]{fud15a}
Drew Fudenberg, Philipp Strack, and Tomasz Strzalecki.
\newblock Speed, accuracy, and the optimal timing of choices.
\newblock \emph{American Economic Review}, 108:\penalty0 3651--3684, 2018.

\bibitem[Guo and Zervos(2010)]{GZ10}
Xin Guo and Mihail Zervos.
\newblock $\pi$ options.
\newblock \emph{Stochastic Processes and their Applications}, 120\penalty0
  (7):\penalty0 1033--1059, 2010.

\bibitem[H\'{e}bert and Woodford(2017)]{heb17}
Benjamin H\'{e}bert and Michael Woodford.
\newblock Rational inattention with sequential information sampling.
\newblock Working paper, Stanford University and Columbia University, 2017.

\bibitem[Ishii(1987)]{ishii1987perron}
Hitoshi Ishii.
\newblock Perron's method for {H}amilton-{J}acobi equations.
\newblock \emph{Duke Mathematical Journal}, 55\penalty0 (2):\penalty0 369--384,
  1987.

\bibitem[Ishii(1989)]{ishii1989uniqueness}
Hitoshi Ishii.
\newblock On uniqueness and existence of viscosity solutions of fully nonlinear
  second-order elliptic pde's.
\newblock \emph{Communications on Pure and Applied Mathematics}, 42\penalty0
  (1):\penalty0 15--45, 1989.

\bibitem[Johnson(1987)]{joh87}
H.~Johnson.
\newblock Options on the maximum or the minimum of several assets.
\newblock \emph{Journal of Financial Quantitative Analysis}, 22\penalty0
  (3):\penalty0 277--283, 1987.

\bibitem[Karatzas and Shreve(1991)]{KS}
Ioannis Karatzas and Steven Shreve.
\newblock \emph{Brownian motion and stochastic calculus}, volume 113.
\newblock Springer-Verlag, New York, second edition, 1991.

\bibitem[Ke and Villas-Boas(2019)]{ke19}
T.Tony Ke and J.Miguel Villas-Boas.
\newblock Optimal learning before choice.
\newblock \emph{Journal of Economic Theory}, 180:\penalty0 383--437, 2019.

\bibitem[Ke et~al.(2016)Ke, Shen, and Villas-Boas]{ke16}
T.Tony Ke, Zuo-Jun~Max Shen, and J.Miguel Villas-Boas.
\newblock Search for information on multiple products.
\newblock \emph{Management Science}, 62\penalty0 (12):\penalty0 3576--3603,
  2016.

\bibitem[Keller and Rady(1999)]{rad99}
Godfrey Keller and Sven Rady.
\newblock Optimal experimentation in a changing environment.
\newblock \emph{Review of Economic Studies}, 66\penalty0 (3):\penalty0
  475--507, 1999.

\bibitem[Kinderlehrer and Stampacchia(1980)]{kinderlehrer1980introduction}
David Kinderlehrer and Guido Stampacchia.
\newblock \emph{An introduction to variational inequalities and their
  applications}, volume~31.
\newblock SIAM, 1980.

\bibitem[Moscarini and Smith(2001)]{mos01}
Giuseppe Moscarini and Lones Smith.
\newblock The optimal level of experimentation.
\newblock \emph{Econometrica}, 69\penalty0 (6):\penalty0 1629--1644, 2001.

\bibitem[Nikandrova and Pancs(2018)]{nik18}
Arina Nikandrova and Roman Pancs.
\newblock Dynamic project selection.
\newblock \emph{Theoretical Economics}, 13:\penalty0 115--144, 2018.

\bibitem[Peskir and Shiryaev(2006)]{PSbook}
Goran Peskir and Albert Shiryaev.
\newblock \emph{Optimal stopping and free-boundary problems}.
\newblock Springer, 2006.

\bibitem[Revuz and Yor(1999)]{RY}
Daniel Revuz and Marc Yor.
\newblock \emph{Continuous martingales and {B}rownian motion}, volume 293.
\newblock Springer-Verlag, Berlin, third edition, 1999.

\bibitem[Roberts and Weitzman(1981)]{rob81}
Kevin Roberts and Martin~L. Weitzman.
\newblock Funding criteria for research, development, and exploration projects.
\newblock \emph{Econometrica}, 49\penalty0 (5):\penalty0 1261--1288, 1981.

\bibitem[Rubinstein(1991)]{rub91}
M.~Rubinstein.
\newblock Somewhere over the rainbow.
\newblock \emph{Risk}, 4:\penalty0 63--66, 1991.

\bibitem[Strulovici and Szydlowski(2015)]{str15}
Bruno Strulovici and Martin Szydlowski.
\newblock On the smoothness of value functions and the existence of optimal
  strategies in diffusion models.
\newblock \emph{Journal of Economic Theory}, 159:\penalty0 1016--1055, 2015.

\bibitem[Stulz(1982)]{stu82}
R.M. Stulz.
\newblock Options on the minimum or the maximum of two risky assets: analysis
  and applications.
\newblock \emph{Journal of Financial Economics}, 10\penalty0 (2):\penalty0
  161--185, 1982.

\bibitem[Wald(1945)]{wal45}
Abraham Wald.
\newblock Sequential tests of statistical hypotheses.
\newblock \emph{Annals of Mathematical Statistics}, 16\penalty0 (2):\penalty0
  117--186, 1945.

\end{thebibliography}
\end{document}